\newcommand{\CLASSINPUTtoptextmargin}{0.74in}
\newcommand{\CLASSINPUTbottomtextmargin}{0.96in}
\documentclass[conference]{IEEEtran}
\IEEEoverridecommandlockouts
\usepackage{cite}
\usepackage{amsmath,amssymb,amsfonts,amsthm}
\usepackage{algorithmic}
\usepackage{graphicx}
\usepackage{textcomp}
\usepackage{xcolor}
\usepackage{url}
\usepackage{lipsum}
\usepackage{multirow}
\usepackage[normalem]{ulem}
\usepackage{bbding}
\usepackage{pifont}
\usepackage{wasysym}
\usepackage{bbm}
\usepackage{dsfont}
\usepackage[ruled,lined,boxed]{algorithm2e}
\usepackage{mathtools, nccmath}
\usepackage{makecell}

\DeclareMathAlphabet{\mathbcal}{OMS}{cmsy}{b}{n}
\def\BibTeX{{\rm B\kern-.05em{\sc i\kern-.025em b}\kern-.08em
		T\kern-.1667em\lower.7ex\hbox{E}\kern-.125emX}}

\renewcommand{\baselinestretch}{0.94}
\newtheorem{prop}{Proposition}

\begin{document}
	
	\title{Low-Complexity Near-Field Localization with XL-MIMO Sectored Uniform Circular Arrays\\
	}
	\author{
		\IEEEauthorblockN{
			Shicong~Liu~
			and Xianghao~Yu
		}
		\IEEEauthorblockA{
			Department of Electrical Engineering, City University of Hong Kong, Hong Kong}
		Email: sc.liu@my.cityu.edu.hk, alex.yu@cityu.edu.hk
	}
	
	\maketitle
	
	\begin{abstract}
		Rapid advancement of antenna technology catalyses the popularization of extremely large-scale multiple-input multiple-output (XL-MIMO) antenna arrays, which pose unique challenges for localization with the inescapable near-field effect. In this paper, we propose an efficient near-field localization algorithm by leveraging a sectored uniform circular array (sUCA). In particular, we first customize a backprojection algorithm in the polar coordinate for sUCA-enabled near-field localization, which facilitates the target detection procedure. We then analyze the resolutions in both angular and distance domains via deriving the interval of zero-crossing points, and further unravel the minimum required number of antennas to eliminate grating lobes. The proposed localization method is finally implemented using fast Fourier transform (FFT) to reduce computational complexity. Simulation results verify the resolution analysis and demonstrate that the proposed method remarkably outperforms conventional localization algorithms in terms of localization accuracy. Moreover, the low-complexity FFT implementation achieves an average runtime that is hundreds of times faster when large numbers of antenna elements are employed.
	\end{abstract}
	
	
	\section{Introduction}
	\bstctlcite{IEEEexample:BSTcontrol}
	In recent years, the advancement of antenna technology has promoted the miniaturization and integration process of extremely large-scale multiple-input multiple-output (XL-MIMO) antenna arrays, which is envisioned as one of the core techniques in the sixth-generation (6G) wireless communication systems~\cite{10068140}. The vast aperture of XL-MIMO arrays not only improves spatial multiplexing and beamfocusing gain for communications~\cite{10379539} 
	but also enables advanced sensing capabilities such as target localization and tracking~\cite{10286475}. However, the increased aperture drastically extends the near-field region, where the planar wavefront assumptions no longer hold. In addition to the angular domain, 
	another spatial degree of freedom (DoF), i.e., distance, needs to be taken into consideration, which poses substantial challenges for localization problems. 
	
	To tackle the near-field localization problem, methods based on uniform linear arrays (ULAs) have been extensively studied\cite{ren2023sensingassisted,10379539,10149471,9707730,10286475}. The conventional multiple signal classification (MUSIC) algorithm was applied to locate near-field user equipment (UE) and scatterers~\cite{ren2023sensingassisted}, while a successive grid search over two DoFs, namely, distance and angle, is necessitated in near-field localization. 
	To reduce the computational complexity, Pan {\it et al.} put forward to decouple the distance and angular parameters and estimate them separately~\cite{10149471}. However, 
	the cubically increased complexity of singular value decomposition (SVD) still cannot be alleviated. A radar sensing method was also utilized to detect the target with orthogonal frequency division multiplexing (OFDM) waveform~\cite{9707730}, which significantly reduces the computational complexity. 
	
	Nevertheless, the ambiguity function of radar sensing algorithms on ULAs shows significant asymmetry with varying impinge azimuth angles, which causes inconsistent resolution on the angular domain~\cite{1589932}, and can easily lead to localization error for objects in close proximity.
	To tackle this issue, uniform circular arrays (UCAs) were identified to be capable of providing azimuthally consistent resolution~\cite{9520393}, which makes them excellent candidates for localization applications. 

Equipped with UCAs, existing works achieved localization by codebook design for the beam training procedure. 
A non-uniform angular and distance grid planning scheme with slight overlapping was proposed in~\cite{10005200}, which produces an elliptical focusing pattern in the near-field region. Later on, an alternate sampling method in both angle and distance domains was proposed to design the near-field concentric-ring codebook~\cite{10243590}. 
Yet, owing to the constraint of codebook size, these methods are trading localization accuracy for efficiency. 
Furthermore, it is unlikely for UCAs to serve UEs through all radiating elements. In practice, sectors of radiating elements are commonly deployed, which have been less studied in the near field. 
Hence, leveraging sectored UCAs, accurate localization algsorithms with low computational complexity in the near field remains an open issue.

In this paper, we address the near-field localization issue under the sectored UCA (sUCA) antenna configuration and OFDM waveform. We propose to estimate the locations of UE and scatterers 
by developing a backprojection algorithm in the polar coordinate. The resolution of the proposed algorithm in both distance and angular domains is further analyzed with accurate approximations.
We also derive the minimum number of antennas required to achieve the aforementioned resolutions. 
It is proved that both the resolution and minimum number of antennas are proportional to the carrier frequency, the radius of the sUCA, and the sector angular span. Finally, 
the computational complexity of the proposed algorithm is reduced by implementing the proposed method with fast Fourier transform (FFT). Numerical results are consistent with the performance analysis and demonstrate that our proposed method significantly outperforms the conventional MUSIC algorithm in both localization accuracy and computational complexity.

\par {\it Notations}: We use normal-face letters to denote scalars and lowercase (uppercase) boldface letters to denote column vectors (matrices). The $k$-th row vector and the $m$-th column vector of matrix ${\bf H}\in\mathbb{C}^{K\times M}$ are denoted as ${\bf H}[{k,:}]$ and ${\bf H}[{:,m}]$, respectively, and the $n$-th element in the vector $\bf h$ is denoted by ${\bf h}[n]$. $\{{\bf H}_n\}_{n=1}^N$ denotes a matrix set with the cardinality of $N$. The superscripts $(\cdot)^{T}$, $(\cdot)^{\rm *}$, and $(\cdot)^{H}$ represent the transpose, conjugate, and conjugate transpose operators, respectively. $\mathcal{CN}(\mu,\sigma^2)$ denotes the complex Gaussian distribution with mean $\mu$ and standard deviation $\sigma$. $\mathbb{E}[\cdot]$ and $\lceil\cdot\rceil$ denote the statistical expectation and ceiling operators, respectively. 
The imaginary unit is represented as $\jmath$ such that $\jmath^2=-1$.
\section{System Model}
\begin{figure}
	\centering
	\includegraphics[width=0.35\textwidth]{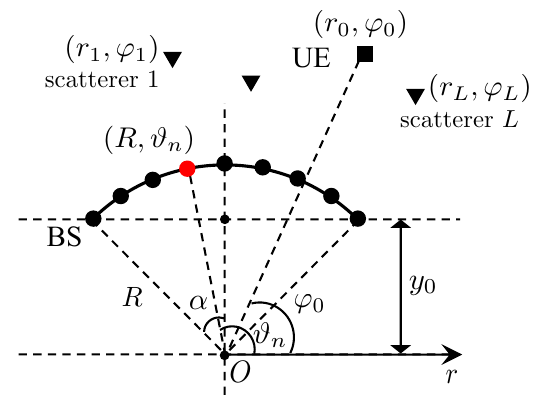}
	\caption{The considered near-field communication scenario.}
	\label{fig:sysmodel}
\end{figure}
As shown in Fig.~\ref{fig:sysmodel}, a single antenna UE is communicating with a BS in its near-field region, where the BS is equipped with a continuous sector of a UCA spanning radians $2\alpha<\pi$ with $N$ antenna elements. The radius of BS is $R$, and the center of the virtual circle of the sUCA is denoted $O$. The vertical distance between the leftmost (or rightmost) element on the sUCA and the $r$-axis of the polar coordinate is denoted by $y_0 = R\cos\alpha$.
We consider an uplink localization scenario, where OFDM wideband modulation with $K$ subcarriers and $f_{\rm SCS}$ subcarrier spacing (SCS) is adopted, and therefore the bandwidth is $f_{\rm BW} = Kf_{\rm SCS}$. The received signal on the $k$-th ($k\in\{1,\cdots,K\}$) subcarrier at the BS can be expressed as
\begin{equation}
	{\bf y}_k = {\bf h}_k{s}_k+{\bf n}_k,
	\label{eq:sysmodel}
\end{equation}
where $s_k$ and ${\bf n}_k\sim {\mathcal{CN}}(0,\sigma_n^2)$ represent the pilot signal and the additive white Gaussian noise (AWGN) on the $k$-th subcarrier, respectively, and $\sigma_n^2$ denotes the noise power. ${\bf h}_k\in\mathbb{C}^{N\times 1}$ is the uplink channel vector and can be modeled using the spherical wave transmission model~\cite{8755476} as
\begin{equation}
	\begin{aligned}
		{\bf h}_k[n] &= \frac{1}{\sqrt{NL}}\sum_{\ell=0}^L  \frac{\alpha_\ell e^{\jmath \frac{2\pi}{c}f_k \left( d_\ell+ d_{\ell,n} \right)}}{ d_{\ell} d_{\ell,n}}\\
		&=\sum_{\ell=0}^L  \frac{\tilde{\alpha}_{k,\ell} e^{\jmath \frac{2\pi}{c}f_k   d_{\ell,n} }}{ d_{\ell,n}},
		\label{eq:channel}
	\end{aligned}
\end{equation}
where $c$ is the light speed in vacuum, $L$ is the number of scatterers, and $\alpha_\ell\sim\mathcal{CN}(0,1)$ is the complex attenuation of the $\ell$-th path. $\tilde{\alpha}_{k,0} = 1$ and $\tilde{\alpha}_{k,\ell} = \frac{\alpha_\ell e^{\jmath \frac{2\pi}{c}f_kd_\ell}}{d_\ell\sqrt{NL}}$ ($\ell\in\{1,\cdots, L\}$) denote the effective channel attenuation for line-of-sight (LoS) and non-LoS (NLoS) paths, respectively. Furthermore, $d_\ell$ denotes the distance from the UE to the $\ell$-th scatterer, and the distance from the $\ell$-th scatterer to the $n$-th element on the BS array is defined by
\begin{equation}
	d_{\ell,n} = \sqrt{r_{\ell}^2+R^2-2Rr_{\ell} \cos\left(  \vartheta_n-\varphi_\ell \right) }.
	\label{eq:distance1}
\end{equation}
$f_k = f_c+ \frac{k}{K} f_{\rm SCS}$ denotes the frequency of the $k$-th subcarrier, while $f_c$ is the carrier frequency.

\section{Proposed Near-Field Localization Method}
\label{sec:prop}
In this section, we propose a backprojection algorithm to reconstruct the near-field scattering environment for localization. In particular, by emulating the reverse propagation of the received signal~\eqref{eq:sysmodel} from the BS back to the locations of the UE and scatterers, location coordinates can be obtained by detecting the peaks in the reconstruction grids. 
Without loss of generality, we assume that the pilot symbols satisfy $\vert s_k \vert = 1$ for all $k\in\{1,\cdots,K\}$, and the emulated signal, also known as the \textit{ambiguity function}, at arbitrary coordinate $(r,\varphi)$ at the $k$-th subcarrier is given by
\begin{equation}
	\begin{aligned}
		F_{k}(r,\varphi)
		={}& \sum_{n=1}^{N} {\bf y}_k[n]s_k^* e^{-j\frac{2\pi}{c}f_k d_n}\\
		={}& \sum_{\ell=0}^{L}\sum_{n=1}^{N}\frac{\tilde{\alpha}_{k,\ell}e^{\jmath \frac{2\pi}{c}f_k \left( d_{\ell,n}-d_n \right) }}{ d_{\ell,n}} + \tilde{\bf n}_k,
	\end{aligned}
	\label{eq:tiori}
\end{equation}
where $\tilde{\bf n}_k = \sum_{n=1}^N {\bf n}_k s_k^* e^{-j\frac{2\pi}{c}f_k d_n}$ is the effective noise, and $d_n$ shares a similar form as~\eqref{eq:distance1} denoting the distance from the polar coordinate $(r,\varphi)$ to the $n$-th antenna element. The multiplied phase term $e^{-j\frac{2\pi}{c}f_k d_n}$ in~\eqref{eq:tiori} denotes exactly the phase shift during the inverse transmission from the BS to $(r,\varphi)$, which emulates the reverse propagation of electromagnetic (EM) wave. 
The binomial expansion of distance term $d_{\ell,n}$ in~\eqref{eq:distance1} is
\begin{equation}
	\begin{aligned}
		d_{\ell,n}\approx r_\ell +\frac{R^2}{2r_\ell}-R\cos\left( \vartheta_n - \varphi_\ell \right),
		\label{eq:binomial}
	\end{aligned}
\end{equation}
and the equality asymptotically holds when $r_\ell^2 \gg  R^2-2 R r_\ell \cos\left( \vartheta_n - \varphi_\ell \right) $. Substituting the binomial expansion~\eqref{eq:binomial} into both distance terms, i.e, $d_{\ell,n}$ and $d_{n}$, in
\eqref{eq:tiori} yields
\begin{align}
	&F_k\left(r,\varphi \mid r_{\ell},\varphi_\ell\right)\notag\\
	{\approx}&
	\sum_{\ell=0}^{L}\sum_{n=1}^{N}\!\frac{e^{-\jmath \frac{2\pi}{c}f_k R \left( \cos\left( \vartheta_n - \varphi_\ell \right)-\cos\left( \vartheta_n - \varphi \right) \right) }}{ r_{\ell}+\frac{R^2}{2 r_{\ell}}-R\cos\left( \vartheta_n-\varphi_\ell \right) }+ \tilde{\bf n}_k\notag\\
	=&\sum_{\ell=0}^{L}\sum_{n=1}^{N}\!\frac{e^{\jmath \frac{4\pi}{c}f_k R\sin\left(\frac{\varphi_\ell-\varphi}{2} \right)\sin\left( \vartheta_n -\frac{\varphi_\ell+\varphi}{2}\right)  }}{ r_{\ell} +\frac{R^2}{2 r_{\ell}}-R\cos\left( \vartheta_n-\varphi_\ell \right) }+ \tilde{\bf n}_k\notag\\
	\triangleq&\sum_{\ell=0}^{L} F_k^{(\ell)}(r,\varphi)+ \tilde{\bf n}_k,
	\label{eq:discrete_form}
\end{align}
from which we need to extract the localization information. Next, we present a favorable property of the constructed function $F_k^{(\ell)}(\cdot)$ in~\eqref{eq:discrete_form}, which facilitates efficient near-field localization.
\begin{prop}
	\label{prop:1}
	There exists only one unique location coordinate $(r,\varphi)$ that maximizes the power of $F_k^{(\ell)}(r,\varphi)$ at $(r,\varphi) = (r_\ell,\varphi_\ell)$.
\end{prop}
\begin{proof}
	The power of the $\ell$-th multi-path component $F_k^{(\ell)}(r,\varphi)$ can be reformulated by Cauchy–Schwarz inequality as
	
	\begin{align}
		&\left\vert F_k^{(\ell)}(r,\varphi)\right\vert^2
		= \left\vert \sum_{n=1}^{N} \frac{e^{\jmath\frac{4\pi}{c}f_k R\left( \cos\left( \vartheta_n - \varphi_\ell \right)-\cos\left( \vartheta_n - \varphi \right) \right) }}{ r_{\ell}+\frac{R^2}{2 r_{\ell}}-R\cos\left( \vartheta_n-\varphi_\ell \right)} \right\vert^2 \notag\\
		\overset{(a)}{\leq}& \sum_{n=1}^{N}  \frac{\left\vert e^{\jmath\frac{4\pi}{c}f_k R \cos\left( \vartheta_n - \varphi_\ell \right) }\right\vert^2\!\! }{\hat{d}_{\ell,n}}
		\sum_{n=1}^{N}  \frac{\left\vert e^{-\jmath\frac{4\pi}{c}f_k R \cos\left( \vartheta_n - \varphi \right) }\right\vert^2\!\! }{\hat{d}_{\ell,n}},\label{eq:rec_pwr}
	\end{align}
	where $\hat{d}_{\ell,n}$ denotes the denominator of $F_k^{(\ell)}(\cdot)$ for notational brevity. 
	The equality of $(a)$ holds if and only if $F_k^{(\ell)}(r,\varphi)$ as $(r,\varphi) = (r_\ell,\varphi_\ell)$, which proves the uniqueness of the power peak.
\end{proof}
Similar to ULA antenna configurations, we can also prove the asymptotic vanishment at other coordinates $(r,\varphi) \neq (r_\ell,\varphi_\ell)$ in $F_k^{(\ell)}(r,\varphi)$ with the presence of $f_k\rightarrow\infty$ and $N\rightarrow\infty$~\cite{liu2024sensingenhanced}, which indicates that the location coordinates can be extracted through a $2$D grid search. However, there is one additional key benefit of adopting the backprojection algorithm on the sUCA. Compared to the ULA case, it is much easier to detect peaks in the reconstruction grid, as exemplified in Fig.~\ref{fig:demo}. In particular, the detection of two targets in Fig.~\ref{fig:demo}(a) requires a 2D grid search with sufficiently large numbers of grids and sophisticated methods. In contrast, with the proposed backprojection algorithm for the sUCA, we can locate the target in the angular domain readily by a simple sum-up operation in the distance dimension and implement a $1$D search for efficient localization. This again verifies the motivation and effectiveness of developing the backprojection algorithm for sUCA-enabled near-field localization.

\begin{figure}
	\centering
	\includegraphics[width=0.4\textwidth]{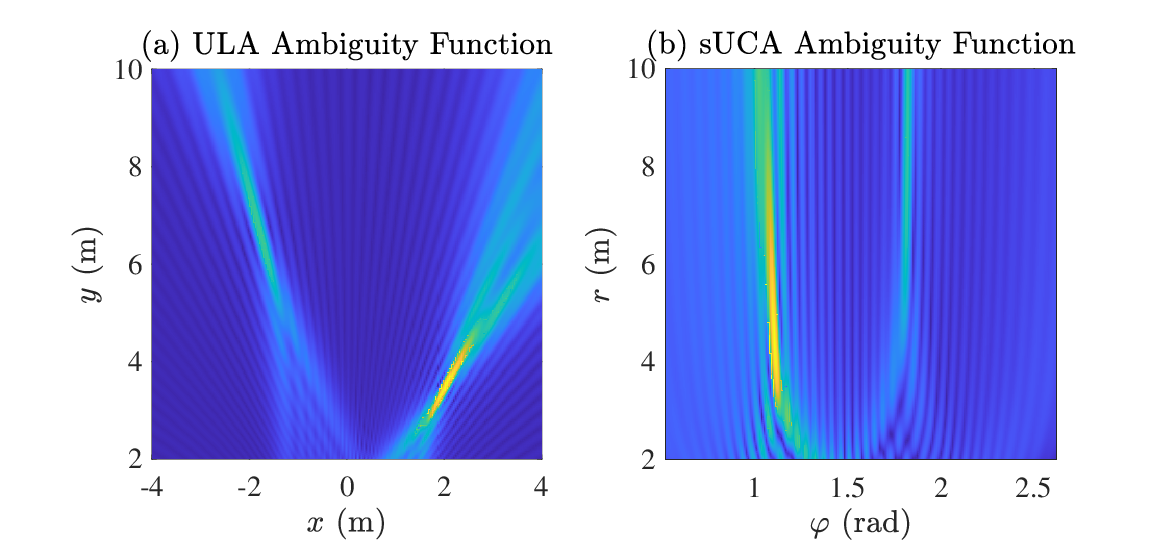}
	\caption{The ambiguity functions of (a) ULA and (b) sUCA with two targets in the near-field region, the (c) $x$-axis projection, and (d) $\varphi$-axis projection of the ambiguity functions.}
	\label{fig:demo}
\end{figure}

\section{Resolution Analysis of the Proposed Method}
Given that the localization parameters can be obtained in a maximum likelihood manner, as presented in Section~\ref{sec:prop}, we need to further derive the resolution in both angular and distance domains to validate the reliability of the proposed algorithm. In this section, we analyze the main lobe widths in the angular and distance domains, respectively, which are directly related to the resolution of localization~\cite{10.5555/1111206}. We also provide the minimum number of antenna elements to achieve the given resolutions.
\subsection{Angular Resolution}
\label{sec:ang}
We first analyze the reconstruction pattern in the angular domain. Let $\hat{F}_k^{(\ell)}\left(\varphi \mid r = r_{\ell},\varphi_\ell\right)$ be the integral form approximation~\cite{9536436} of the discrete sum of the $\ell$-th multi-path component $F_k^{(\ell)}(r,\varphi)$ in~\eqref{eq:discrete_form}, where the approximation error approaches $0$ when $N\rightarrow\infty$. 
The angular pattern at $r = r_{\ell}$ is then derived as
\begingroup
\allowdisplaybreaks
\begin{align}
	&\left\vert\hat{F}_k^{(\ell)}\left(\varphi \mid r = r_{\ell},\varphi_\ell\right)\right\vert \notag\\
	={}&\left\vert\int_{\frac{\pi}{2}-\alpha}^{\frac{\pi}{2}+\alpha} \frac{e^{\jmath \frac{2\pi}{c}f_k R \left[ \cos\left( \theta-\varphi_\ell \right) - \cos\left( \theta-\varphi \right) \right]  }}{ d_{\ell}(\theta)}~{\rm d}\theta\right\vert \notag\\
	\overset{(b)}{\approx}{}&\left\vert\int_{\frac{\pi}{2}-\alpha}^{\frac{\pi}{2}+\alpha} \frac{e^{\jmath z \sin\left( \theta -\frac{\varphi_\ell+\varphi}{2}\right)  }}{ {r}_{\ell}}~{\rm d}\theta\right\vert\label{eq:angle}\\
	={}&\frac{1}{r_\ell}\left\vert J_0\!\left( \frac{\pi-\varphi_\ell-\varphi}{2}\!\!+\!\alpha,z \right)\!\!+\!\! J_0^*\!\left( \frac{\pi-\varphi_\ell-\varphi}{2}\!\!-\!\alpha,z \right) \right\vert,\notag
\end{align}
\endgroup
where $d_\ell(\theta)$ denotes the distance from the $\ell$-th target to the sUCA element at $(R,\theta)$, and $(b)$ holds with the same condition as~\eqref{eq:binomial}. $z=\frac{4\pi}{c}f_k R\sin\left(\frac{\varphi_\ell-\varphi}{2} \right)$, and $J_0\left(w,z\right)$ denotes \textit{the first kind incomplete Bessel function of order zero}, with $w$ being the integral upper bound. Since the properties of incomplete Bessel functions in~\eqref{eq:angle} are not yet clear, we need to find the peaks and zero-crossing points to determine the lobe widths and resolution in the angular domain. 

When $\varphi = \varphi_\ell$, as considered in \textbf{Proposition~\ref{prop:1}}, the phase term in~\eqref{eq:angle} becomes $z\sin\left( \theta -\left({\varphi_\ell+\varphi}\right)/{2}\right) = 0$, which leads to in-phase superposition. The maximum peak of $\left\vert\hat{F}_k^{(\ell)}\left(\varphi \mid r = r_{\ell},\varphi_\ell\right)\right\vert$ is therefore obtained at $\varphi=\varphi_\ell$. However, when we have $\varphi \neq \varphi_\ell$, although it can be proved that the angular pattern will finally vanish to $0$, it is difficult to obtain the distance of adjacent zero-crossing points of incomplete Bessel functions. Let $t = \theta-\pi/2$, the integral in~\eqref{eq:angle} can be given by
\begin{equation}
	\begin{aligned}
		g(\Delta\varphi) &= \int_{\frac{\pi}{2}-\alpha-\varphi_\ell}^{\frac{\pi}{2}+\alpha-\varphi_\ell} \frac{e^{\jmath \frac{2\pi}{c}f_k R \left[ \cos\left( \theta-\varphi_\ell \right) - \cos\left( \theta-\varphi \right) \right]  }}{ r_{\ell}}~{\rm d}t\\
		&\overset{(c)}{\approx} \int_{-\alpha-\varphi_\ell}^{\alpha-\varphi_\ell} \frac{e^{-\jmath \frac{2\pi}{c}f_k R \Delta\varphi \cos\left( t-\varphi_\ell \right) }}{ r_{\ell}}~{\rm d}t,
	\end{aligned}
	\label{eq:angle2}
\end{equation}
where $\Delta\varphi = \varphi-\varphi_\ell$, and the approximation $(c)$ holds when 
$\vert\Delta\varphi\vert$ is small. To estimate the distance of adjacent zero-crossing points, we can choose to utilize the distance between adjacent local maxima and minima as a substitute. Hence, we can set the angular span in the phase term of~\eqref{eq:angle2} to be the integer multiple of $2\pi$. Particularly, for $\varphi_\ell = \pi/2$ case\footnote{For cases when $\varphi_\ell\neq\pi/2$, the range of $\cos\left( t-\varphi_\ell \right)$ lacks monotonicity in the considered integral domain, which only slightly deteriorates the lobe widths with practical system parameters.}, we have 
\begin{equation}
	\Delta\varphi = \frac{m\lambda_k}{2R\sin\alpha},
	\label{eq:lobewidth}
\end{equation}
where $\lambda_k = c/f_k$ denotes the wavelength of the $k$-th subcarrier, and integer $m$ is the index of the local maxima and minima. With $m=1$, the resolution of the angular pattern, represented by the width of the main lobe, is proportional to the wavelength and inversely proportional to the radius $R$, and the sine value of the sector span $\sin\alpha$.

\subsection{Distance Resolution}
We then determine the distance resolution for the proposed localization algorithm. Similar to~\eqref{eq:angle}, the distance reconstruction pattern at $\varphi = \varphi_\ell$ is given by

\begingroup
\allowdisplaybreaks
\begin{align}
		&\left\vert\hat{F}_k^{(\ell)}\left(r \mid r_{\ell},\varphi = \varphi_\ell\right)\right\vert \notag\\
		={}&\left\vert\int_{\frac{\pi}{2}-\alpha}^{\frac{\pi}{2}+\alpha} \frac{e^{\jmath \frac{2\pi}{c}f_k\left( r_\ell+\frac{R^2}{2r_\ell} - r-\frac{R^2}{2r} \right)  }}{ r_\ell +\frac{R^2}{2 r_\ell}-R\cos\left( \theta-\varphi_\ell \right) }~{\rm d}\theta\right\vert \notag\\
		={}&\left.\frac{2{\rm arctanh}\left[ \sqrt{\frac{\left( R+r_\ell \right)^2+r_\ell^2}{{\left( R-r_\ell \right)^2+r_\ell^2}}} \tan\left(\frac{\varphi_\ell-\theta}{2} \right) \right]}{\sqrt{r_\ell^2+\frac{R^2}{4r_\ell^4}}}\right\vert_{\frac{\pi}{2}-\alpha}^{\frac{\pi}{2}+\alpha}=C,
\end{align}
\endgroup
which is a constant. Hence, different from linear arrays, the sUCA does not show distance resolution with one single subcarrier. 
Thanks to the OFDM waveform, we can utilize $K$ subcarriers with spacing $f_{\rm SCS}$ to improve the distance resolution. By summing the distance patterns calculated on all $K$ subcarriers, we can obtain
\begingroup
\allowdisplaybreaks
\begin{align}
		\left\vert \sum_{k=1}^{K}\hat{F}_k^{(\ell)}\left(r \mid r_\ell,\varphi = \varphi_\ell\right) \right\vert&=\left\vert C\sum_{k=1}^{K} e^{\jmath \frac{2\pi}{c}f_k\left( r_\ell+\frac{R^2}{2r_\ell} - r-\frac{R^2}{2r} \right)  } \right\vert\notag\\
		&=\left\vert C \frac{\sin\left( \frac{\pi f_{\rm SCS}}{c} K\Delta r \right)}{\sin\left( \frac{\pi f_{\rm SCS}}{c}\Delta r \right)}\right\vert,
	\label{eq:lobedis}
\end{align}
\endgroup
where $\Delta r = r_\ell+\frac{R^2}{2r_\ell} - r-\frac{R^2}{2r}$. Therefore, the width of the main lobe is $2c/(Kf_{\rm SCS}) = 2c/f_{\rm BW}$. 

\subsection{Minimum Number of Antennas}
\begin{figure}
	\centering
	\includegraphics[width=0.4\textwidth]{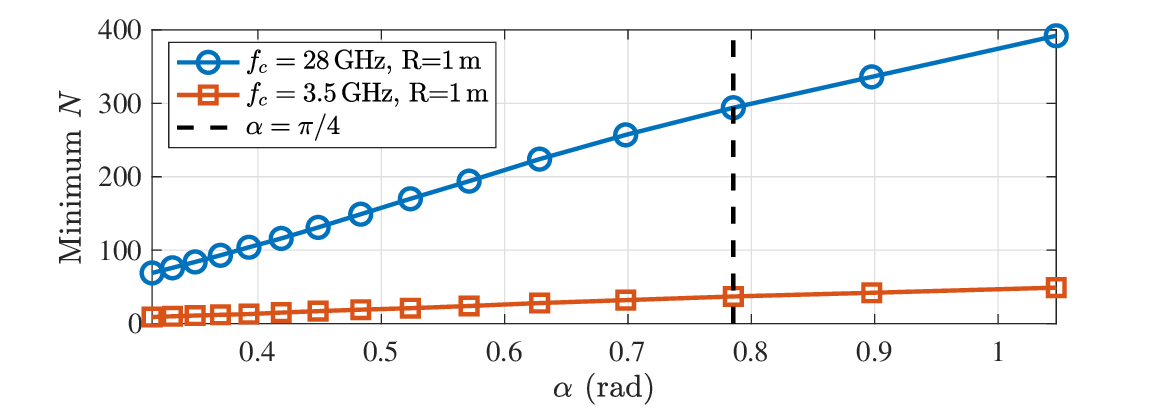}
	\caption{Minimum number of antennas according to~\eqref{eq:minantenna} for $\pi/4\leq \alpha < \pi/2$ and~\eqref{eq:minantenna2} for $0< \alpha < \pi/2$.}
	\label{fig:ant}
	\vspace{-3mm}
\end{figure}
As mentioned in Section~\ref{sec:ang}, the approximated integral form in~\eqref{eq:angle} shows negligible approximation error as $N\rightarrow \infty$, 
while practically it is required to find the minimum number of antennas for lower hardware complexity and costs. Similar to ULAs, to avoid grating lobes with the sUCA antenna configuration, the largest phase difference between the EM waves received at two adjacent antenna elements should be smaller than $\pi$~\cite{10.5555/1111206}. For the typical case when $\pi/4<\alpha<\pi/2$, (i.e., the array sector covers a region larger than $\pi/2$ radians but no more than $\pi$), we need to ensure 
\begin{equation}
	\sin\frac{\alpha}{N} \leq \frac{\lambda}{4R}.
	\label{eq:case1}
\end{equation}
Substituting $\sin ({\alpha}/{N})\simeq \alpha/N$ (as $\alpha\ll N$) into~\eqref{eq:case1}, the minimum number of antennas should satisfy
\begin{equation}
	N\geq \left\lceil \frac{4\alpha R}{\lambda}\right\rceil,
	\label{eq:minantenna}
\end{equation}
which is closely related to the reciprocal of the lobe width. On the contrary, when we have $\alpha<\pi/4$, the minimum antenna spacing similarly needs to satisfy $R\cos\left( 2\alpha -2\alpha/N \right)-R\cos\left( 2\alpha \right)\leq \lambda /2$, which further yields
\begin{equation}
	N\geq \left\lceil\frac{2\alpha}{2\alpha - \arccos\left( \frac{\lambda}{2R}+\cos2\alpha \right)}\right\rceil.
	\label{eq:minantenna2}
\end{equation}
Fig.~\ref{fig:ant} illustrates the minimum number of antennas $N$ with a typical array radius $R=1\,$m and carrier frequencies $f_c\in\{3.5,28 \}\,$GHz.


\section{Low-Complexity Implementation}

The proposed algorithm can be reformulated using convolution operations and implemented by FFT to significantly reduce the computational complexity. In particular, \eqref{eq:tiori} can be written as
\begin{equation}
	\begin{aligned}
		&F_{k}\left(\varphi\mid r, r_{\ell},\varphi_\ell\right)
		= \sum_{n=1}^{N} {\bf y}_k[n]s_k^* e^{-j\frac{2\pi}{c}f_k d_n^\prime}\\
		={}& \sum_{n=1}^{N} \overline{\bf y}_k[n] e^{-j\frac{2\pi}{c}f_k \sqrt{   r^2+R^2-2Rr \cos\left( \vartheta_n - \varphi \right)}}.
		\label{eq:conv1}
	\end{aligned}
\end{equation}
Note that for a given distance ${r}=r^{(j)}$, which corresponds to the $j$-th distance grid out of a total of $G_d$ grids, \eqref{eq:conv1} can be reformulated as the linear convolution between vector $\overline{\bf y}_k = {\bf y}_k s_k^{*} \in\mathbb{C}^{N\times 1}$ and 
\begin{equation}
	\begin{aligned}
		\overline{\bf e}_{r^{(j)}}  &= \left[ e^{\jmath \frac{2\pi f_k}{c}\sqrt{ \tilde{r}-2Rr \cos\left( \varphi^{(1)} \right)}}, \cdots,\right.\\ 
		&\quad\quad~~\left.e^{\jmath \frac{2\pi f_k}{c}\sqrt{ \tilde{r} - 2Rr \cos\left( \varphi^{(G_a)} \right)}} \right],
	\end{aligned}
\end{equation}
where $\tilde{r} = (r^{(j)})^2 +R^2$, and $G_a>N$ is the number of uniform angle grids. The linear convolution between $\overline{{\bf y}}_k$ and $\overline{\bf e}_{r^{(j)}}$ can be further implemented by FFT as
\begin{equation}
	\boldsymbol{\Theta}_k[:,j]= \mathcal{F}_{G_a^\prime}^{-1}\left[  \mathcal{F}_{G_a^\prime}\left[ \overline{\bf y}_k \right]\odot \mathcal{F}_{G_a^\prime}\left[ \overline{\bf e}_{r^{(j)}} \right]  \right]_{N:G_a^\prime-N+1},
	\label{eq:fft}
\end{equation}
where $\mathcal{F}_a\left[ \cdot \right]$ and $\mathcal{F}^{-1}_a\left[ \cdot \right]$ denote FFT and inverse FFT (IFFT) operations with padding (or trimming) length $2^{\lceil \log a \rceil}$, respectively, and $G_a^\prime = G_a+2N$ is used to compensate for the edge effects to ensure at least $G_a$ valid points after convolution. By employing FFT in the convolution computation, the complexity can be reduced from $\mathcal{O}\left( NG_a \right)$ to $\mathcal{O}\left( P\log P \right)$, where $P = 2^{\lceil \log G_a \rceil}$ is the nearest power-$2$ number that is greater than $G_a$.
With the low-complexity implementation, the proposed localization algorithm is summarized in \textbf{Algorithm~\ref{alg:ti}}.

\begin{algorithm}[!t]
	\caption{Low-Complexity Implementation}\label{alg:ti}
	\begin{algorithmic}[1]
		\REQUIRE Received signal $\{{\bf y}_k \}_{k=1}^K$, reconstruction grids $G_a$ on the angular domain and $G_d$ on the distance domain, the number of targets $L$ to be detected\footnotemark, array sector span $\alpha$, and distance grid range $[r_{\rm min},r_{\rm max}]$.
		\ENSURE The estimated coordinates of the UE and scatterers.
		\FOR{$k = 1,\cdots,K$}
		\STATE Initialize empty matrix $\boldsymbol{\Theta}_k\in\mathbb{R}^{G_a\times G_d}$.
		\FOR{$j = 1,\cdots, G_d$}
		\STATE Fill the $j$-th column of $\boldsymbol{\Theta}_k$ according to~\eqref{eq:fft}.
		\ENDFOR
		\STATE Normalize $\boldsymbol{\Theta}_k$ with
		\vspace{-2mm}
		\begin{equation}
			\boldsymbol{\Theta}_k:=\boldsymbol{\Theta}_k/\max\left( \left\vert\boldsymbol{\Theta}_k\right\vert \right).\notag
			\vspace{-2mm}
		\end{equation}
		\ENDFOR
		\STATE Find top $L+1$ indices of peaks  $\{i_\ell\}_{\ell=0}^L$ from 
		\vspace{-2mm}
		\begin{equation}
			\boldsymbol{\Theta}_a = \sum_{k=1}^{K}\sum_{j=1}^{G_d}\left\vert \boldsymbol{\Theta}_k [:,k]\right\vert.\notag
			\vspace{-2mm}
		\end{equation}
		\STATE Find peak on distance grids from
		\vspace{-2mm}
		\begin{equation}
			j_\ell = \arg\underset{j}{\max} \left\vert\sum_{k=1}^{K}  \boldsymbol{\Theta}_k[i_\ell,j] \right\vert.\notag
			\vspace{-2mm}
		\end{equation}
		\STATE Angle estimations $\hat{\varphi}_\ell = \pi/2-\alpha+2\alpha i_\ell/G_a$.
		\STATE Distance estimations $\hat{r}_\ell = r_{\rm min}+(r_{\rm max}-r_{\rm min})j_\ell/G_d$.
	\end{algorithmic}
\end{algorithm}

\section{Simulation Results}
In this section, we evaluate the accuracy and the computational complexity of the proposed localization method via numerical simulations. The accuracy is measured by the localization error defined by,
\begin{equation}
	r_{\rm err} = \frac{1}{L+1}\sum_{\ell=0}^{L}\sqrt{ \hat{r}^2_\ell+r_\ell^2 - 2\hat{r}_\ell r_\ell\cos\left( {\hat \varphi}_\ell - \varphi_\ell \right) } ,
\end{equation} 
where $({\hat r}_\ell,{\hat \varphi}_\ell)$ is an estimation of $({r}_\ell,{\varphi}_\ell)$, while the complexity is evaluated by the number of multiplications involved.
\subsection{Simulation Setup}
\begin{figure}[t]
	\centering
	\includegraphics[width=0.4\textwidth]{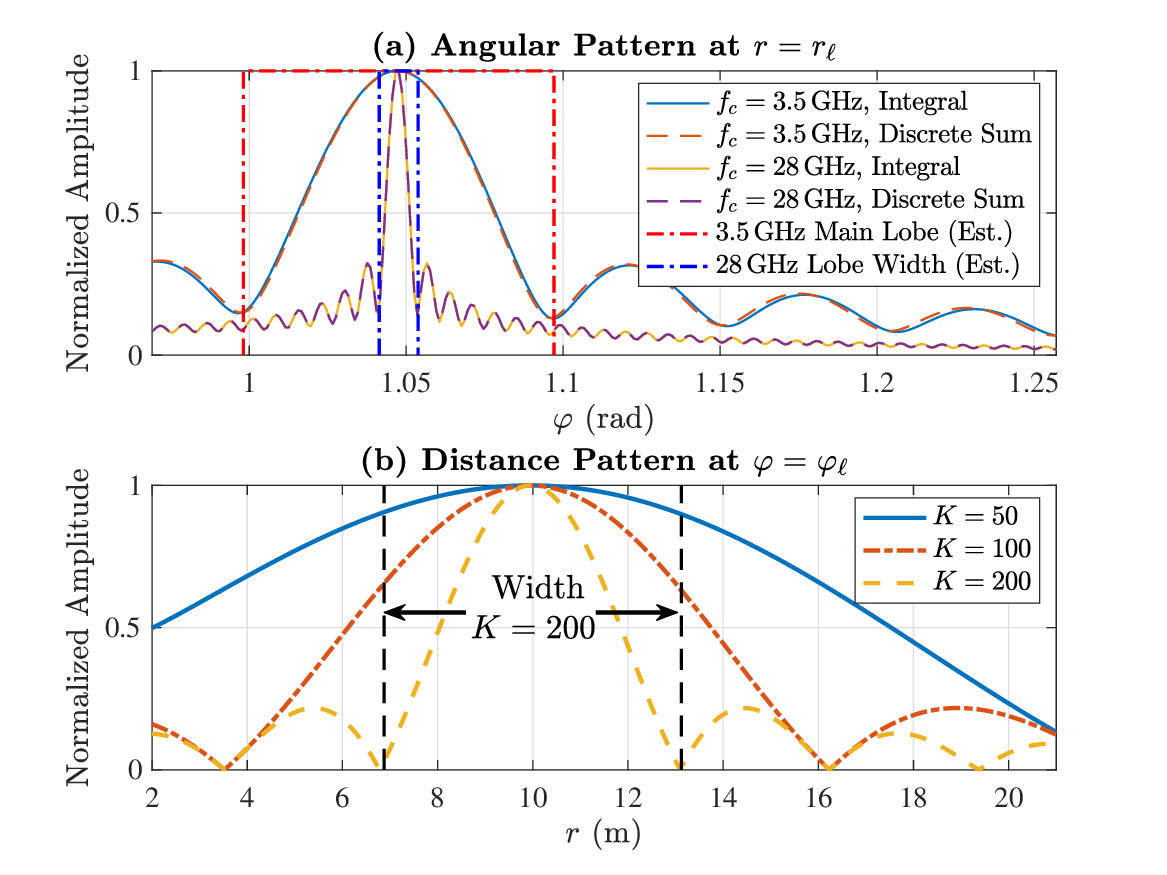}
	\caption{The reconstruction pattern of the proposed algorithm in the (a) angular and (b) distance domains, and the corresponding lobe widths estimated according to~\eqref{eq:lobewidth} and~\eqref{eq:lobedis}, respectively.}
	\label{fig:approx}
	\vspace{-4mm}
\end{figure}

In the simulation, we consider the angle of the sector as $2\alpha = 2\pi/3$, and the radius of the sUCA is set as $R=1\,$m unless otherwise stated. The sUCA at BS is equipped with the minimum number of antennas according to~\eqref{eq:minantenna}. The carrier frequency is set to $f_c=3.5\,$GHz for the sub-$6$G frequency range and $f_c=28\,$GHz for the millimeter-wave frequency range, respectively, and the SCS for OFDM waveforms is set as $f_{\rm SCS} = 480\,$kHz.\footnotetext{To obtain the value of $L$, an empirical threshold can be adopted to identify the peaks in $\boldsymbol{\Theta}$ that are significantly higher than other ripples.}\footnote{According to the technical specifications of 3GPP, SCS up to $960\,$kHz is supported in Release 17~\cite{3gpp.38.300}.} The location coordinate of UE is selected uniformly from $2\,$m to $21\,$m within the $2\alpha$ sector area, which is strictly inside the near-field region.

In both the MUSIC algorithm and the FFT implementation of the proposed algorithm, we set the same reconstruction grids as $G_a = 2N$ and $G_d=100$ in the angular and distance domains, respectively.

\subsection{Numerical Results}

We first visualize the reconstruction pattern in the near-field to verify the theoretical results. The $\ell$-th target is fixed at polar coordinate $(r_\ell,\varphi_\ell)=(10\,{\rm m},\pi/3)$, and the reconstruction patterns are calculated given $r=r_\ell$ and $\varphi = \varphi_\ell$, respectively. 
Fig.~\ref{fig:approx}(a) and (b) show the reconstruction pattern on angular and distance domains, respectively, with normalized amplitudes. Specifically, for the angular domain, we use the solid line to plot the approximated integral form results in~\eqref{eq:angle} and the dashed lines to plot the discrete sum in~\eqref{eq:discrete_form}. The approximation shows negligible error, which confirms the rationality and accuracy of the adopted approximations. Additionally, the dotted-dashed lines plot the estimated lobe widths according to~\eqref{eq:lobewidth}, and the vertical lines reveal the accuracy of the estimated lobe widths with carrier frequencies of $f_c = 3.5\,$GHz and $f_c = 28\,$GHz. The width of the main lobe becomes narrower as the carrier frequency becomes higher, which is consistent with the analytical result.

In Fig.~\ref{fig:approx}, the distance reconstruction pattern is also calculated according to~\eqref{eq:discrete_form}, and the dashed line marks the estimated lobe width according to~\eqref{eq:lobedis} with $K=200$ subcarriers. The simulation results also align with the analysis that the lobe width is inversely proportional to the bandwidth $f_{\rm BW} = Kf_{\rm SCS}$.

\begin{figure}[t]
	\centering
	\includegraphics[width=0.4\textwidth]{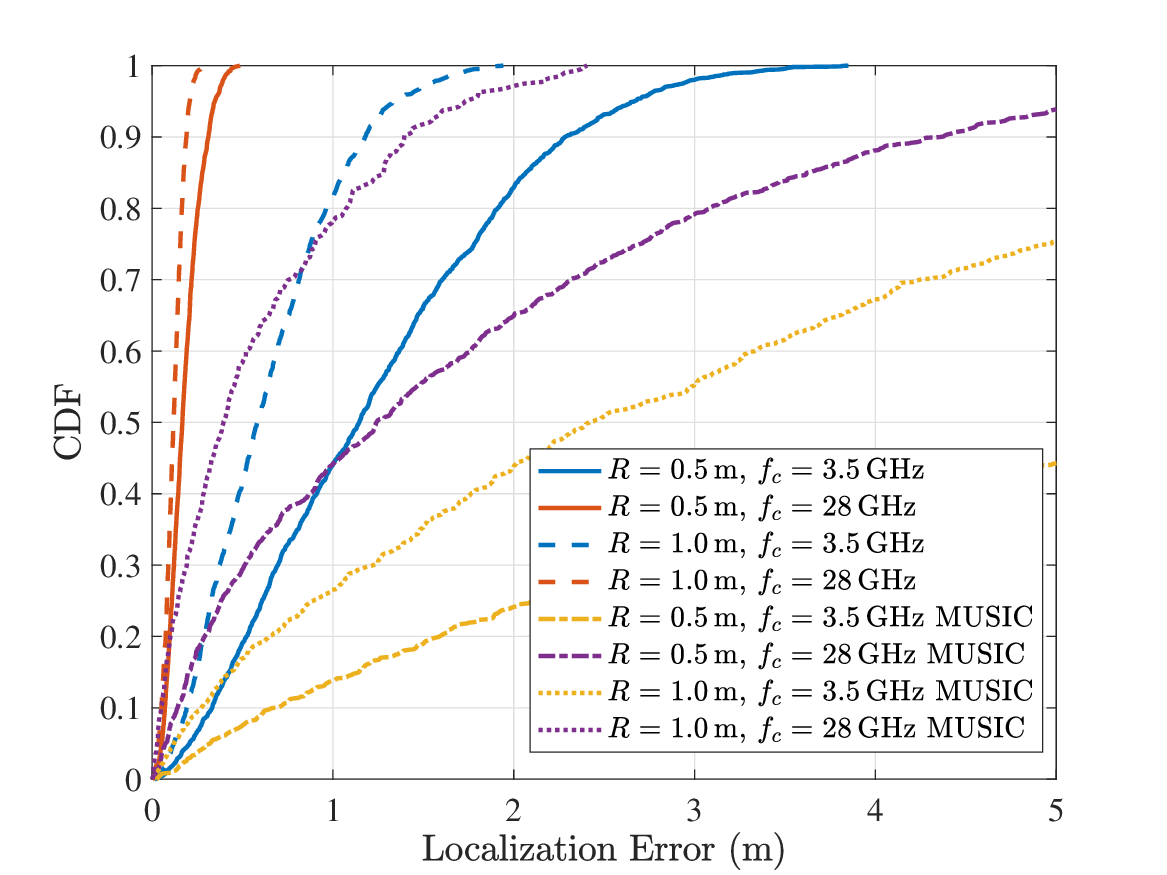}
	\caption{The empirical CDF of the localization error of the proposed method and MUSIC algorithm.}
	\label{fig:cdf}
	\vspace{-4mm}
\end{figure}

The localization accuracy is then investigated with two different array radii $R\in\{0.5\,{\rm m},1\,{\rm m}\}$ in terms of the cumulative distribution function (CDF) of the localization error. As can be observed from Fig.~\ref{fig:cdf}, the localization error is reduced with the increased array radius and carrier frequency, which aligns with the derived resolution~\eqref{eq:lobewidth} in the angular domain. It is worth noting that the MUSIC algorithm performs unsatisfactorily in all scenarios, which can be explained by three-fold reasons. First, the steering matrix of sUCA no longer satisfies Vandermonde form, and therefore localization problems cannot be directly formulated as an eigenvalue decomposition (EVD) problem, for which the MUSIC algorithm excels. Second, since the two DoFs are mutually coupled in the near-field steering vector, estimating either parameter individually will result in significant errors. Finally, the original MUSIC algorithm and some of its variations are not asymptotically consistent estimators when the numbers of antennas and samples increase~\cite{4400832}. In contrast, by emulating the reverse transmission of the EM waves, our proposed backprojection algorithm for sUCA-enabled localization shows a substantial reduction in both mean and standard deviation of localization errors.

The computational complexity in terms of the number of multiplications and average runtime is further evaluated in Table~\ref{tab:loc}. The average runtime is simulated in $f_c\in\{3.5,28\}\,$GHz scenarios, which correspond to $N=49$ and $N=392$ antennas, respectively. 
Given that $G_a$ is proportional to the number of antennas $N$ in the simulation setup, the number of multiplications increases quartically with $N$, while the proposed method increases quadratically. Consequently, the proposed method achieves up to a hundred times faster than the MUSIC algorithm with $f_c = 28\,$ GHz carrier frequency and $N = 392$ antenna elements.

\begin{table}[t]
	\setlength{\tabcolsep}{2pt}
	\centering
	\caption{Computational Complexities}
	\vspace{-2mm}
	\label{tab:loc}
	\begin{tabular}{c|c|cc}
		\hline\hline
		\multirow{2}{*}{} &  \multirow{2}{*}{Number of Multiplications} & \multicolumn{2}{c}{Avg. Runtime (s)} \\ \cline{3-4} 
		&                        & $3.5\,$GHz      & $28\,$GHz      \\ \hline
		MUSIC                                 & \makecell{$KN^2+N^3+$\\$G_aG_d\left( (N-L+1)N^2+N \right)$}                         &   $0.5071$                          &  $58.8952$                                \\ \hline
		Proposed                         &  $KG_d(G_a+1)N$    &   $0.5963$                           &  $0.7995$                              \\ \hline
		Proposed (FFT)         &  $KG_d\lceil \log_2 G_a \rceil 2^{\lceil \log_2 G_a \rceil}$      & $\bf 0.0787$              &  $\bf 0.5879$                   \\ \hline\hline
	\end{tabular}
\end{table}
\begin{figure}
	\centering
	\includegraphics[width=0.4\textwidth]{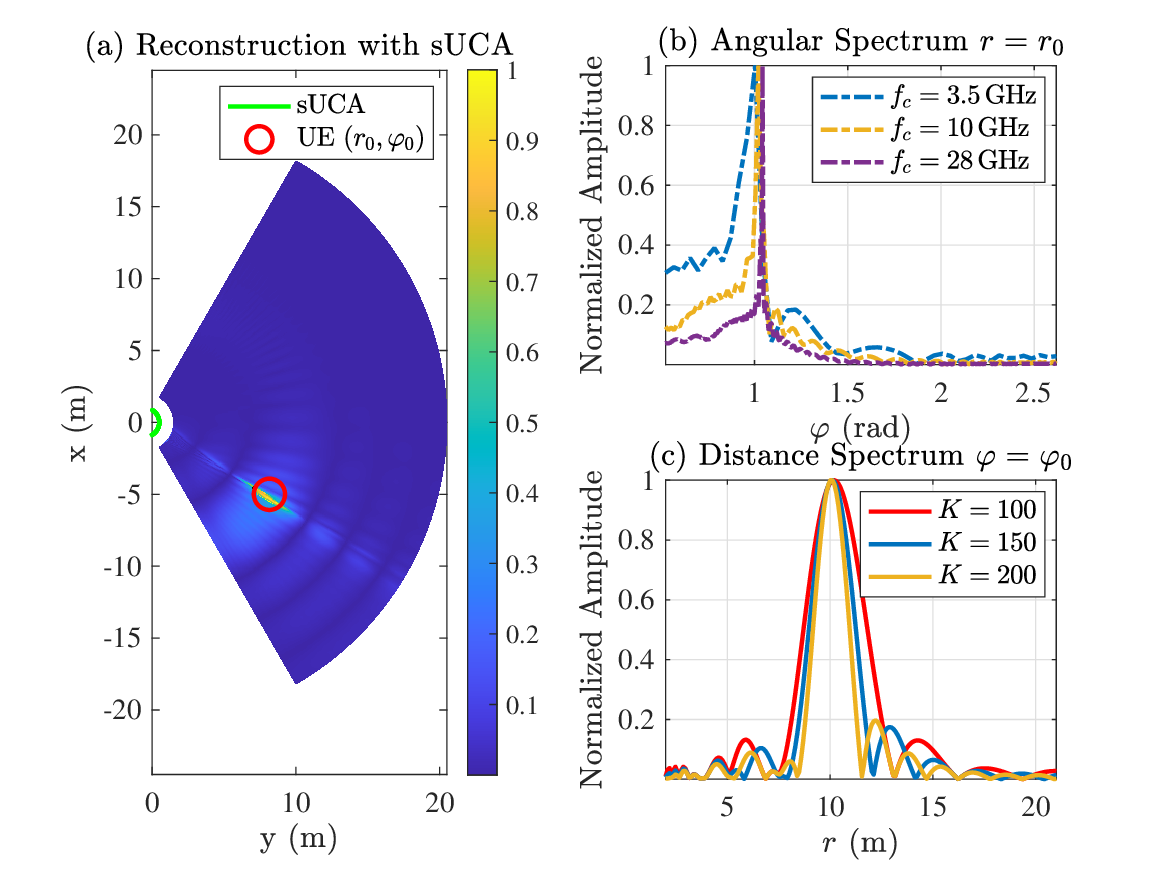}
	\caption{(a) The reconstructed near-field environment with a single UE via FFT implementation, (b) the angular domain reconstruction pattern at $r = r_0$, and (c) the distance domain reconstruction pattern at $\varphi=\varphi_0$.}
	\label{fig:suca}
	\vspace{-4mm}
\end{figure}

We finally investigate the performance of the proposed algorithm implemented by FFT. For demonstration purposes, we set the carrier frequency as $f_c \in\{3.5, 10, 28\}\,$GHz, and consider up to $K=200$ subcarriers. Fig.~\ref{fig:suca}(a) shows the reconstructed near-field region in a Cartesian coordinate system, where the target can be accurately localized by the proposed method. We further plot the reconstruction results in both angular and distance domains, respectively, in Figs.~\ref{fig:suca}(b) and (c) when $r=r_0$ and $\varphi=\varphi_0$. Since the FFT implementation takes finite grids $G_a$ in the angular domain, the oscillation, as previouly observed in Fig.~\ref{fig:approx}(a) in the angular domain, becomes insignificant in Fig.~\ref{fig:suca}(b), while the widths of the main lobes still shrink inversely proportional to the increase of carrier frequency as expected. Similarly, in the distance domain, as depicted in Fig.~\ref{fig:suca}(c), the lobe width also shrinks as the number of subcarriers increases, indicating that wider bandwidth results in higher distance resolution.

\section{Conclusions}

In this paper, we addressed the localization issue in the near-field region with a sUCA antenna configuration by resorting to the backprojection algorithm. We demonstrated that the physical architecture of sUCA enables the proposed method to operate directly in the polar coordinate system, which promotes efficient and low-complexity localization tasks. We further revealed the resolution of the proposed method in both angular and distance domains. Specifically, the angular resolution is inversely proportional to the carrier frequency, the radius of sUCA, and the sine value of the sector angle span, while the distance resolution is inversely proportional to the signal bandwidth. We finally implemented the proposed method by FFT, which achieves up to a hundred times faster than the MUSIC algorithm in XL-MIMO OFDM systems.

\bibliographystyle{IEEEtran}
\bibliography{IEEEabrv,references}

%

\end{document}